\documentclass[a4paper]{article}

\usepackage[round]{natbib}
\usepackage[colorinlistoftodos]{todonotes}
\usepackage[english]{babel}
\usepackage[utf8]{inputenc}
\usepackage{enumitem}
\usepackage{amsmath}
\usepackage{complexity}
\usepackage{graphicx}
\usepackage{tikz}
\usetikzlibrary{shapes,arrows,snakes,backgrounds,calendar}
\tikzstyle{vert}=[circle,draw,minimum size=12pt,inner sep=0pt, text=black!50,line width=0.5pt]
\tikzstyle{edge} = [-]

\usepackage[linesnumbered,ruled,vlined,]{algorithm2e}

\SetCommentSty{mycommfont}

\newcommand{\ie}{\emph{i.e.}}
\newcommand{\eg}{\emph{e.g.}}

\newcommand{\eq}{\leftarrow}

\usepackage{amsthm}
\newtheorem{proposition}{Proposition}
\newtheorem{lemma}{Lemma}

\title{Maximum-expectation matching under recourse}

\author{João Pedro Pedroso and Shiro Ikeda}

\date{May 24, 2016}

\begin{document}
\maketitle

\begin{abstract}
This paper addresses the problem of maximizing the expected size of a matching in the case of unreliable vertices and/or edges.  The assumption is that upon failure, remaining vertices that have not been matched may be subject to a new assignment.  This process may be repeated a given number of times, and the objective is to end with the overall maximum number of matched vertices.

The origin of this problem is in kidney exchange programs, going on in several countries, where a vertex is an incompatible patient-donor pair; the objective is to match these pairs so as to maximize the number of served patients.  A new scheme is proposed for matching rearrangement in case of failure, along with a prototype algorithm for computing the optimal expectation for the number of matched vertices.  

Computational experiments reveal the relevance and limitations of the algorithm, in general terms and for the kidney exchange application.
\end{abstract}

\section{Introduction}
\label{sec:introd}

Algorithms for matching have recently raised interest as a research topic on a particular application: maximizing the number of transplants in kidney exchange programs.  These programs have been organized in several countries, in order to provide patients with kidney failure with an alternative to traditional treatment, in cases where they have a donor willing to provide a kidney, but the pair is not physiologically compatible~\citep{Klerk05,Biro09,Saidman06,Segev05}. These programmes are based on the concept of \emph{exchange} between two patient-donor pairs: donors are allowed to provide a kidney to the other pair's patient, if compatibility exists, so that both patients benefit.

Figure \ref{fig:swap} (left) illustrates the simplest case with only two pairs, $(P_1, D_1)$ and $(P_2, D_2)$.  Donor $D_1$ of the first pair is allowed to give a kidney to patient $P_2$ of the second pair, and patient $P_1$ may get a kidney from donor $D_2$. These graphs concern only preliminary compatibilities, which must be reassessed prior to actual transplant, by confirming the availability of the intervening persons and through additional medical exams.

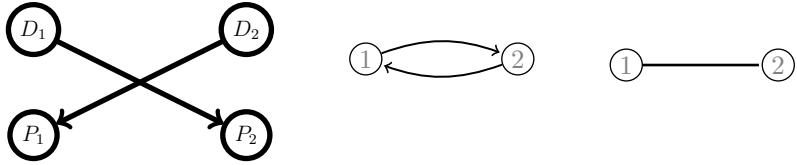
\begin{figure}[!htbp]
  \centering
  \hfill{}
  \begin{tikzpicture}[baseline={(current bounding box.center)}, thick, scale=.7, transform shape]
    \path 
    (-2.0, 1.0) node (1) [line width=2pt,shape=circle,draw,label=below:] {\large{$D_1$}}
    (-2.0,-1.0) node (2) [line width=2pt,shape=circle,draw,label=below:] {\large{$P_1$}} 
    ( 2.0, 1.0) node (3) [line width=2pt,shape=circle,draw,label=above:] {\large{$D_2$}} 
    ( 2.0,-1.0) node (4) [line width=2pt,shape=circle,draw,label=left :] {\large{$P_2$}} 
    ;
    \draw [line width=2pt, ->] (1) to (4);
    \draw [line width=2pt, ->] (3) to (2);
  \end{tikzpicture}
  \hfill{}
  \begin{tikzpicture}[shorten >=1pt, scale=1.0, line width=0.75pt, ->]
    \foreach \name/\x/\y in {1/0/0, 2/2/0}
    \node[vert] (G-\name) at (\x,\y) {$\name$};
    
    \foreach \from/\to in {1/2,2/1}
    { \draw  (G-\from) to [bend left=20] (G-\to); }
  \end{tikzpicture}
  \hfill{}
  \begin{tikzpicture}[shorten >=1pt, scale=1.0, line width=1pt]
    \foreach \name/\x/\y in {1/0/0, 2/2/0}
    \node[vert] (G-\name) at (\x,\y) {$\name$};
    
    \begin{scope}[every node/.style={scale=.667}]
      \path (G-1) edge [right] (G-2);
    \end{scope}
  \end{tikzpicture}
  \hfill{}
  \caption{An exchange between two incompatible pairs; arcs represent a preliminary assessment compatibility, and arrows define a possible exchange (left).  Representation of this situation as a directed graph (center), and as a graph where an edge stands for a pair of opposite arcs between two nodes (right).}
  \label{fig:swap}
\end{figure}

The situation depicted in Figure~\ref{fig:swap} may be extended to cycles with more than two vertices; typical programs may allow two to five vertices in an exchange, though smaller cycles are preferable due hospital logistics, and other reasons.  The optimization problem underlying a standard kidney exchange usually considers maximizing the number of transplants~\citep{Klerk05,Segev05}, though other criteria have been proposed; \eg, the overall weight assigned to each transplant~\citep{Li2014,Manlove2015}, and the expected number of transplants for selected cycles~\citep{pedroso2014lncsCOA} or subsets of vertices~\citep{KPV}.  The deterministic version of this problem, where all the elements of the graph are assumed reliable, is naturally modeled through integer optimization; a summary of models for the kidney exchange problem has been presented and analyzed by~\citet{Constantino2013}.

Throughout this paper we will only consider the case of cycles of two vertices; in such a case, the directed graph representing an instance may be simplified, by replacing opposite arcs between two nodes by an edge, and ignoring all the other arcs.  The relevant problem is maximum-weighted matching in a graph.  

Algorithms for finding a maximum matching in a graph have been proposed in the 1960's by~\citet{Edmonds1965a}, based on some properties of maximal matchings by~\citet{Berge1957}.  Maximum-weighted matching in a graph has been studied in~\citet{Edmonds1965b}, who proposed an algorithm involving a formulation of the problem as a linear program, linear programming duality, and the previous algorithm for maximum matching.  The complete algorithm is polynomial, solving the weighted matching problem in $O(n^4)$ time.  An analysis of efficient algorithms and data structures for this problem is presented in~\citet{Galil1986}.

Here we study this problem under uncertainty, in a situation where vertices and/or edges may fail.  This problem has been initially raised in~\citet{Li2011} and \citet{chen2012}, which propose a simulation system for maximizing the expected utility when arcs are subject to failure, in dynamic version of the problem.  Reconfiguration of the solution after a failure is observed has been addressed in~\citep{Manlove2015}, where the number of effective 2-cycles is taken into account in the weights assigned to 3-cycles. Maximization of the expectation of the number of transplants, considering internal reconfiguration within a cycle's vertices, has been considered in~\citet{pedroso2014lncsCOA}.  Reconfiguration involving vertices outside a cycle, in what has been called \emph{subset recourse}, was considered in ~\citet{KPV}.

Here, we take a more general view: we consider that after a failure, any recourse solution may be chosen as long as it does not involve previously matched vertices.  Recourse may be repeated an undetermined number of times, until the \emph{residual graph} has no edges available.  This is relevant in practice, as there are no natural obstacles limiting recourse to a single trial.

This problem seems to be inherently intractable; we provide some examples, discuss how to specify and compute a solution, and analyze the behavior of the proposed algorithm with some instances of a real-world situation where it can be applied, the kidney exchange problem.

Our contributions are the following.  In Section~\ref{sec:descr} we model and analyze matching with repeated recourse.  An algorithm for tackling this problem is provided in Section~\ref{sec:alg}, and its behavior is experimentally assessed in Section~\ref{sec:results}.  We conclude with Section~\ref{sec:conclusions}, which presents considerations on the applicability and limitations of the approach proposed.

\section{Preliminaries and problem description}
\label{sec:descr}

The input to our problem is a graph $G = (V,E)$.  A matching $M$ in $G$ is a subset of $E$ such that no two edges in $M$ have a vertex in common.  A maximal matching $M$ has the property that no edge can be added to $M$ (\ie, $M$ is not a proper subset of any other matching in $G$). 

Vertices $i \in V$ may fail with probability $p_i$, and edges $\{i,j\}$ (also denoted~$ij$) may fail with probability $p_{ij}$.  For the sake of simplicity, except if otherwise stated we will assume that only edges fail.  

In our setting, after a matching is proposed, the matched edges and vertices are \emph{observed}; a matched edge \emph{succeeds} if neither the edge nor the incident vertices fail.  If a matched edge succeeds, the incident vertices are ``served'' and the edge will no longer be considered.  If a matched edge fails, the incident vertices are not served; the failure is permanent, and hence that edge is no longer considered either.  The search for \emph{a posteriori} attempts to serve more vertices is called \emph{recourse}; it is similar to the original problem, except that is may not involve vertices that have been served already.

After each observation, the graph is updated into a \emph{residual graph}.  Edges are eliminated at each observation as follows:
\begin{itemize}
\item successful edges, and all edges at their endvertices, are removed;
\item edges that failed are removed;
\item edges at vertices that failed are removed.
\end{itemize}
If a proposed matching is maximal and in the observation there are no failures, the residual graph will contain no edges.

The objective is to find a matching such that the \emph{expectation} for its cardinality (or weight) is maximum.  The basis for the efficiency of algorithms for maximum-weighted matching is that a maximum matching is also \emph{maximal}.  Unfortunately, this property does not hold in our context.

\begin{proposition}
  Under recourse, a maximum expectation matching may be non-maximal.
\end{proposition}

\begin{proof}
  By counterexample.  Consider the  graph $C_4$, where the labels on edges correspond the their probability of failure.
  \begin{center}
    \begin{tikzpicture}[shorten >=1pt, scale=1.0, line width=1pt]
      \foreach \name/\x/\y in {1/0/0, 2/0/1, 3/1/1, 4/1/0}
      \node[vert] (G-\name) at (\x,\y) {$\name$};
      
      \begin{scope}[every node/.style={scale=.667}]
      \path (G-1) edge [left]   node {$p_{12}$} (G-2)
            (G-2) edge [above]  node {$p_{23}$} (G-3)
            (G-3) edge [right]  node {$p_{34}$} (G-4)
            (G-4) edge [below]  node {$p_{14}$} (G-1);
      \end{scope}
    \end{tikzpicture}
  \end{center}

  Assuming unitary weights on each edge, the expression of the expectation for the weight of the matching $\{\{1,2\}, \{3,4\}\}$ is the following:
  \begin{alignat*}{27}
    & 4 (1-p_{12}) (1-p_{34}) + 2 (1-p_{12}) p_{34} + 2 (1-p_{34}) p_{12} \\
    & + p_{12} p_{34} (4 (1-p_{23}) (1-p_{14}) + 2 (1-p_{14}) p_{23} + 2 (1-p_{23}) p_{14}).
  \end{alignat*}
  The expression of the expectation for the weight of the non-maximal matching $\{\{1,2\}\}$ is the following:
  \begin{alignat*}{27}
    (1-p_{12}) (2 + 2 (1-p_{34})) + p_{12} (& 4 (1-p_{23}) (1-p_{14}) + 2 p_{23} (1-p_{14})\\
                                            & + 2 p_{14} (1-p_{23}) + p_{23} p_{14} (2 (1-p_{34}))).
  \end{alignat*}
  The difference between the former and the latter expressions is:
  $$-2 p_{12} (1 - p_{14}) (1 - p_{23}) (1 - p_{34}),$$
  which is negative or zero.
\end{proof}

Actually, if there is no limit on the number of observations allowed, there is no interest in taking matchings with more than one edge.

\begin{proposition}
  With no limit on the number of observations allowed, there is a maximum-expectation matching with one edge chosen per each observation.
\end{proposition}

\begin{proof}
  Any solution with more than one edge in a given step is still allowed with one edge per observation.
\end{proof}

The previous example puts in evidence that the maximum-expectation matching for that graph is one of $\{\{1,2\}\}$,  $\{\{2,3\}\}$,  $\{\{3,4\}\}$,  or $\{\{1,4\}\}$.  It becomes also clear that the complete specification of a solution involves the statement of the choices in the initial decision, followed by the choices at the second level (after the first observation), and so on, until the set of edges in the residual graph becomes empty.  In this example, it would be (assuming that the optimum is $\{\{1,2\}\}$):
\begin{description}
\item[At level 1:] the matching $\{\{1,2\}\}$;
\item[At level 2:] ~
  \begin{itemize}
  \item If $\{\{1,2\}\}$ succeeds in the observation: the matching $\{\{3,4\}\}$
  \item If $\{\{1,2\}\}$ fails in the observation: one must find the maximum-expectation matching in the residual graph
  \begin{center}
    \begin{tikzpicture}[shorten >=1pt, scale=1.0, line width=1pt]
      \foreach \name/\x/\y in {1/0/0, 2/0/1, 3/1/1, 4/1/0}
      \node[vert] (G-\name) at (\x,\y) {$\name$};
      
      \begin{scope}[every node/.style={scale=.667}]
      \path (G-2) edge [above]  node {$p_{23}$} (G-3)
            (G-3) edge [right]  node {$p_{34}$} (G-4)
            (G-4) edge [below]  node {$p_{14}$} (G-1);
      \end{scope}
    \end{tikzpicture}
  \end{center}
  This graph has four possible matchings:  $\{\{1,4\},\{2,3\}\}$ (maximal), or one of the non-maximal matchings $\{\{1,4\}\}$, $\{\{2,3\}\}$, or $\{\{3,4\}\}$.  It turns out that all of them are equivalent, with expectation
  $$p_{14} (2 + 2 p_{34}) + (1 - p_{14}) (2 - 2 (1 - p_{23}) (1 - p_{34})).$$
  For example, if at level 2 one chooses $\{\{1,4\},\{2,3\}\}$, then level 3 would be
  \begin{itemize}
  \item If both $\{1,4\}$ and $\{2,3\}$ succeed in the observation, then the residual graph will have no edges;
  \item If $\{1,4\}$ succeeds and $\{2,3\}$ fails in the observation, then the residual graph will have no edges;
  \item If $\{1,4\}$ fails and $\{2,3\}$ succeeds in the observation, then the residual graph will have no edges;
  \item If both $\{1,4\}$ and $\{2,3\}$ fail in the observation, then the residual graph will be
  \begin{center}
    \begin{tikzpicture}[shorten >=1pt, scale=1.0, line width=1pt]
      \foreach \name/\x/\y in {1/0/0, 2/0/1, 3/1/1, 4/1/0}
      \node[vert] (G-\name) at (\x,\y) {$\name$};
      
      \begin{scope}[every node/.style={scale=.667}]
      \path (G-3) edge [right]  node {$p_{34}$} (G-4);
      \end{scope}
    \end{tikzpicture}
  \end{center}
  and the obvious choice at level 4 would be the matching $\{\{3,4\}\}$.
\end{itemize}
\end{itemize}
\end{description}

The previous example shows that a complete specification of a solution involves predicting all the possible scenarios, and recursively finding the optimum for each of them.  This problem is not in the nondeterministic polynomial time ($\NP$) complexity class of the underlying maximum-weighted matching problem.  Its recursive nature also prevents its classification in $\PSPACE$, as for a matching with $n$ edges one must check the $2^n$ failure combinations; therefore, it is an intractable problem.  

\begin{lemma}
  Max-expectation matching is not in \EXPSPACE.
\end{lemma}

\begin{proof}
  As there is no advantage in choosing multiple edges simultaneously, we assume that they are chosen sequentially.  The contribution of an edge to the value of the expectation of a matching is not known \emph{a priori}, as it will depend on possible rearrangements with other edges.  Consider a graph with $m$ edges; in the worst case, we will need $m$ steps for selecting a matching.  For each selected edge, we must recursively check what happens if it succeeds and if it fails.  Hence, in the first choice there are $2 m$ cases to analyse; in the second choice, there are $2 (m-1)$ cases, and so on.  The number of steps in the worst case is therefore $2^m m!$, which is not bounded by an exponential in $m$.
\end{proof}

We used the notation of~\citet{papadimitriou1994} for complexity classes.
Given this property, there is little hope of solving large instances. However, an algorithm for tackling this problem may still be useful for small real-world cases.  Besides, limiting the number of observations allowed, as described in the next section, may allow improvements in the actual computational time required.  The actual CPU time needed for solving a set of benchmark instances will be analyzed in Section~\ref{sec:results}.

\subsection{Limited recourse}
\label{sec:limit}

In most situations, there is a limit in the allowed number of observations and recourse reconfigurations.  We call \emph{$N$-recourse} to a matching problem where the solution must be reached within $N$ observations.  If the limit $N$ is zero, the problem falls back to standard matching, maximizing expectation based on failure probabilities, but without recourse; $N=\infty$ corresponds to the unlimited case previously described.  

The difficulty of solving an $N$-recourse problem increases with $N$, being solvable in polynomial time for $N=0$.
Therefore, in terms of complexity, 0-recourse is in \P~(using, \eg, the algorithm proposed by~\citet{Edmonds1965a}), with complexity increasing with $N$, and $\infty$-recourse being intractable, as shown before.

A practical approach to solving this problem consists of obtaining an initial solution for $N=0$, and then incrementing the number of allowed observations $N$ until the additional gain is considered acceptably low, or until the computational time becomes excessive.

\section{Algorithm}
\label{sec:alg}

The basis for the method that we propose is the enumeration of all the matchings in a graph.  An interesting algorithm for enumerating all the minimum-cost perfect matchings (where the vertices in the graph are matched) has been proposed in~\citet{fukuda1992} for the case of bipartite graphs; to the best of our knowledge, there is no equivalent algorithm for more general cases. Algorithm~\ref{alg:matching} proposes a very simple recursive procedure for enumerating all the matchings in a graph.  
\begin{algorithm}[h!tbp]
  \begin{footnotesize}
    \DontPrintSemicolon
    \SetKwFunction{algo}{algo}\SetKwFunction{ConnectedComponents}{ConnectedComponents}
    \SetKwFunction{algo}{algo}\SetKwFunction{EvaluateMatching}{EvaluateMatching}
    \SetKwFunction{algo}{algo}\SetKwFunction{Matchings}{Matchings}
    \SetKwFunction{algo}{algo}\SetKwFunction{Solve}{Solve}
    \KwData{graph:
      \begin{itemize}[nosep]
      \item graph $G=(V,E)$ with edges remaining for enumeration (initially, original graph);
      \item matching $m$ currently under construction (initially empty);
      \item current set of matchings $M$ (initially empty);
      \end{itemize}
    }
    \KwResult{
      \begin{itemize}[nosep]
      \item list of all matchings in $G=(V,E)$.
      \end{itemize}
    }
    \SetKwProg{myproc}{procedure}{}{}
    \SetKw{Break}{break}
    \SetKw{Continue}{continue}
    \myproc{\Matchings{$V, E, m, M$}}{
      \If{$E = \emptyset$}{\Return $M$}
      $ij \eq $ arbitrary edge from $E$\;
      $m' \eq m \cup \{ij\}$\;
      $M \eq M \cup \{m'\}$\;
      $E' \eq \{ab \in E : \{a,b\} \cap \{i,j\} = \emptyset\}$\;
      \Matchings{$V, E', m', M$} \tcp*[f]{case 1: add $ij$ to current matching}\;
      \Matchings{$V, E \setminus \{ij\}, m, M$} \tcp*[f]{case 2: don't add $ij$ to current matching}\;
      \Return{$M$}\;
    }
  \end{footnotesize}
  \caption{Algorithm for enumerating all matchings.}
  \label{alg:matching}
\end{algorithm}

Determining the maximum-expectation matching involves an indirect recursion between the main function \Solve, presented in Algorithm~\ref{alg:solve}, and function \EvaluateMatching, presented in Algorithm~\ref{alg:eval}.  The former starts by finding the connected components present in the graph; as the expectation for each of them is independent of the others, it may be computed separately.  Then, each matching in a given component is evaluated with \EvaluateMatching and the best of them is chosen.
\begin{algorithm}[h!tbp]
  \begin{footnotesize}
    \DontPrintSemicolon
    \KwData{
      \begin{itemize}[nosep]
      \item instance:
        \begin{itemize}[nosep]
        \item set of vertices $V$;
        \item set of edges $E$;
        \item probabilities of failure $p_i$ for vertices $i \in V$ and $p_{ij}$ for edges $ij \in E$;
        \end{itemize}
      \item limit of observations allowed $N$.
      \end{itemize}
    }
    \KwResult{
      \begin{itemize}[nosep]
      \item maximum-expectation value. 
      \end{itemize}
    }
    \SetKwProg{myproc}{procedure}{}{}
    \SetKw{Break}{break}
    \SetKw{Continue}{continue}
    \myproc{\Solve{$V, E, p, N$}}{
      $z^* \eq 0$\;
      \ForEach{$C \in \ConnectedComponents(V,E)$}{
        \lIf{$|C| = 1$}{\Continue}
        $(V',E') \eq $ subgraph induced on vertex set $C$\;
        $z \eq 0$\;
        \ForEach{$m \in \Matchings(V',E')$}{
          $R \eq E'$\;
          $z' \eq \EvaluateMatching(V',E',p,m,R,N)$\;
          \If{$z' \geq z$}{
            $z \eq z'$\;
          }
        }
        $z^* \eq z^* + z$\;
      }
      \Return{$z^*$}\;
    }
  \end{footnotesize}
  \caption{Algorithm for finding the maximum-expectation matching.}
  \label{alg:solve}
\end{algorithm}

The evaluation of a matching involves listing all the patterns of success or failure of its edges.  For each pattern, some bookkeeping is necessary for determining its probability of occurrence; this value is then multiplied by the number of edges for computing the associated contribution to the expectation.  However, edges which did not fail and which were not involved with success in the current pattern (stored in variable $R'$) are free for a rearrangement; this is the reason why \Solve is called inside \EvaluateMatching, for determining their contribution to the expectation under the current pattern.
\begin{algorithm}[h!tbp]
  \begin{footnotesize}
    \DontPrintSemicolon
    \KwData{
      \begin{itemize}[nosep]
      \item instance:
        \begin{itemize}[nosep]
        \item set of vertices $V$;
        \item set of edges $E$;
        \item probabilities of failure for vertices $p_i, i \in V,$ and for egdes $p_{ij}, \forall ij \in E$.
        \end{itemize}
      \item matching $m$;
      \item limit of observations allowed $N$.
      \end{itemize}
    }
    \KwResult{
      \begin{itemize}
      \item expectation considering all failure patterns.
      \end{itemize}
    }
    \SetKwProg{myproc}{procedure}{}{}
    \SetKw{Break}{break}
    \SetKw{Continue}{continue}
    \myproc{\EvaluateMatching{$V, E, p, m, R, N$}}{
      \lIf{$m,N$ was previously memoized}{\Return{$T_{mN}$}}
      $z \eq 0$\;
      \ForEach{$b \in $ binary patterns of size $|m|$}{
        $q \eq 1$\;
        $n \eq 0$\;
        \For{$k \leftarrow 1$ \KwTo $|m|$}{
          $ij \eq k^\text{th}$ edge of matching $m$\;
          \If(\tcp*[f]{edge $ij$ fails in this pattern}){$b_k = 0$}{
            $q \eq q \times p_{ij}$\;
            $R' \eq R \setminus \{ij\}$;
          }
          \Else(\tcp*[f]{edge $ij$ succeeds in this pattern}){
            $q \eq q \times (1-p_{ij})$\;
            $n \eq n + 1$\;
            $R' \eq \{ab \in R : \{a,b\} \cap \{i,j\} = \emptyset\}$\;
          }
        }
        \If{$R' \neq \emptyset$}{
          $z' \eq \Solve(V,R',p,N-1)$\;
          $z \eq z + q \times (2 n + z')$\;
        }
      }
      memoize $T_{mN} \eq z$\;
      \Return{$z$}
    }
  \end{footnotesize}
  \caption{Algorithm for evaluating a matching under recourse.}
  \label{alg:eval}
\end{algorithm}

\section{Results}
\label{sec:results}

Let us start by analyzing what differences may be expected between situations without and with recourse.
Figure~\ref{fig:expectC4K4} shows the expectation for the number of matched vertices for a graph with four edges, for varying probability of failure $p$ on edges (considered identical for all edges).  The graphs considered, with four vertices, are a cycle ($C_4$), and the complete graph ($K_4$).  With no recourse, the expectation is $4(1-p)^2$ for both graphs.  With recourse, we can observe a considerable improvement on graph $C_4$, and further improvements on $K_4$, especially for moderate probability of failure.  This first result motivates for the use of recourse.

\begin{figure}[h!tbp]
  \centering
  \scalebox{0.75}{\input{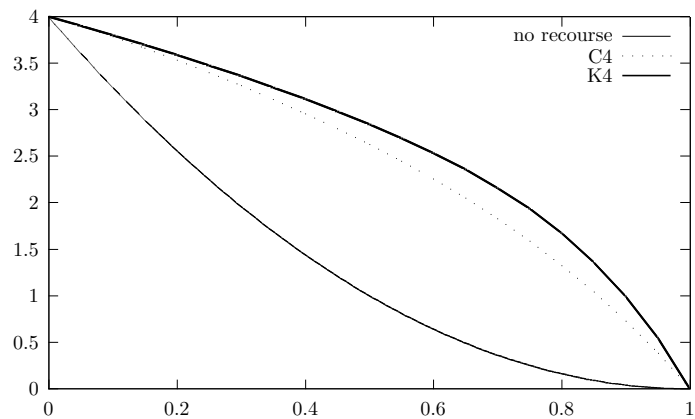}}
  \caption{Expectation for varying probability of failure on edges.}
  \label{fig:expectC4K4}
\end{figure}

\subsection{Results on general graphs}
\label{sec:resgen}

This section reports results for the application of the algorithm in the worst scenario: unlimited number of observations, on possibly dense graphs.  A more realistic experiment is presented in the next section.  As expected, the exponential growth of the time required for solving in terms of the number of edges is clearly observed in Figure~\ref{fig:cpu_graphs}, which plots the CPU time used as a function of the number of edges in the graph, for all graphs with up to six vertices.  The points almost overlap, showing that there is virtually no influence of the number of vertices on the CPU time used; this behavior will no longer be observed in the more realistic cases of next section, where graphs are relatively sparse and have special structure.

\begin{figure}[h!tbp]
  \centering
  \includegraphics[width=.40\textwidth,trim=0 0 0 0,clip=True]{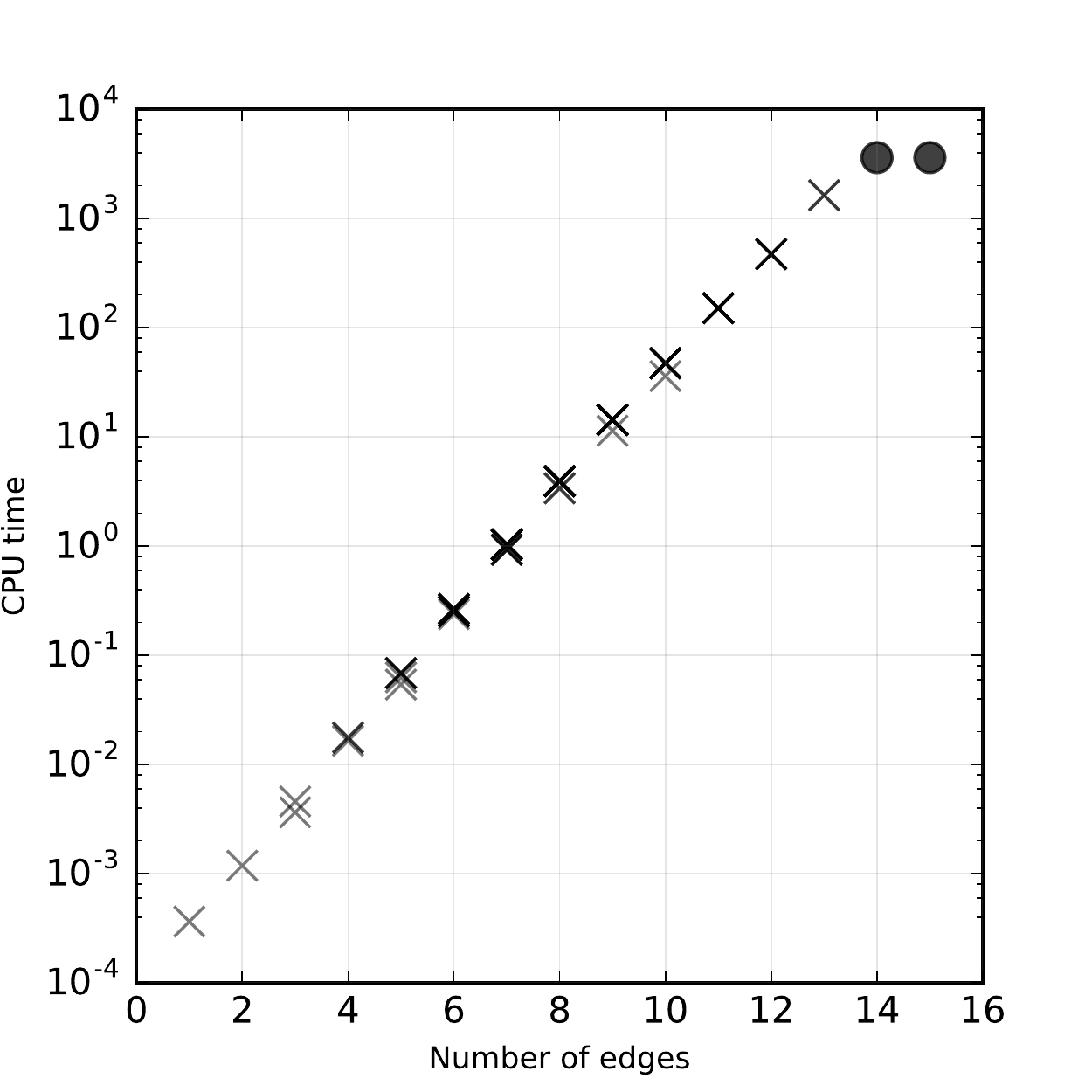}
  \caption{Average CPU time required for solving instances for all graphs up to six vertices, as a function of the number of edges in the graph.  Circles represent cases with no solution within 3600s.}
  \label{fig:cpu_graphs}
\end{figure}

\subsection{Results on an application: kidney exchange}
\label{sec:resapp}

We now turn to the usage of the algorithm proposed on its main application, \ie, in kidney exchange programs, using data provided in~\citet{Constantino2013}.  These data's graphs have been randomly generated based on characteristics of real-world kidney exchange pools; probabilities have been randomly drawn with uniform distribution~\citep{pedroso2014lncsCOA}\footnote{Instances' data are available at \texttt{http://www.dcc.fc.up.pt/\~{}jpp/code/KEP}.}.  These data are for the more general case of directed graphs $G'=(V,A)$, where exchanges may involve more than two pairs.  Undirected graphs $G=(V,E)$ are generated as mentioned above, with an edge $\{i,j\} \in E$ for each 2-cycle $(i,j)-(j,i)$ in the arc set~$A$.  For each edge $\{i,j\} \in E$ we consider its probability of failure as $p_{ij} = p'_i p'_j p'_{ij} p'_{ji}$, where $p'$ are the original probabilities for the directed graph.

Figure~\ref{fig:cpu_kep} shows the time required for solving each benchmark instance, for a number of pairs in the pool varying from 10 to 50; the actual number of vertices in the graph is typically smaller, after removing isolated vertices.  Each point corresponds to a benchmark instance, representing the CPU time required for solving the it in terms of the number of vertices and edges in the graph, for a maximum number $N$ of observations ranging from $N=0$ (no recourse, top) to $N=\infty$ (no limit, bottom).

\begin{figure}[htbp]
  \centering
  \includegraphics[width=.95\textwidth,trim=55 55 60 60,clip=True]{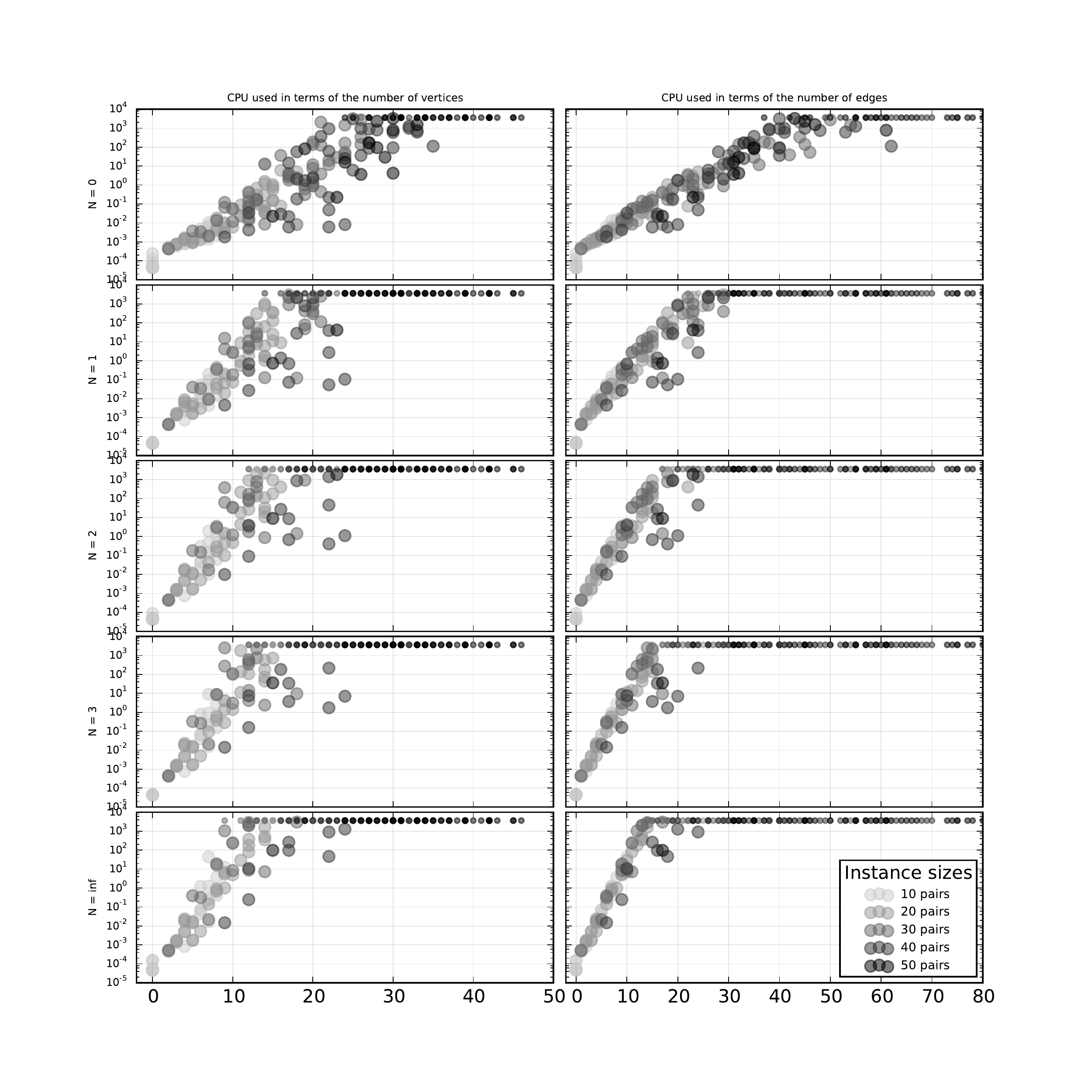}
  \caption{CPU time required for solving realistic KEP instances terms of the number of vertices (left) and edges (right) in the graph, for a maximum number of observations ranging from 0 (top) to no limit (bottom).  Smaller dots represent cases with no solution within 3600s.}
  \label{fig:cpu_kep}
\end{figure}

Table~\ref{tab:results_succ} reports the number of successes, out of 50, for each instance size.  As expected, increasing the number of allowed observations leads to a sharp raise on the CPU time required for solving an instance.  Notice that when the number of observations allowed is zero, the problem can be easily solved for much larger instances with Edmonds algorithm, or even with a general-purpose mixed-integer optimization solver; as our implementation relies on enumerating all the matchings, it is inappropriate in this case.

\begin{table}[htbp]
  \centering
  \begin{tabular}[top]{lrrrrr}
Instances	& \multicolumn{1}{c}{N=0}	& \multicolumn{1}{c}{N=1}	& \multicolumn{1}{c}{N=2}	& \multicolumn{1}{c}{N=3}	& \multicolumn{1}{c}{N=inf} \\ \hline
10 pairs	& 50	& 50	& 50	& 50	& 50 \\
20 pairs	& 50	& 46	& 37	& 32	& 26 \\
30 pairs	& 49	& 30	& 18	& 16	& 14 \\
40 pairs	& 35	& 14	& 11	& 9	& 8 \\
50 pairs	& 13	& 2	& 2	& 1	& 1 \\
  \end{tabular}
  \caption{Number of instances of each size (out of 50) successfully solved.}
  \label{tab:results_succ}
\end{table}

\section{Conclusions}
\label{sec:conclusions}

Dealing with unreliability is as issue with great importance in many applications involving optimization in graphs.  In this paper we introduced a problem that arises in kidney exchange programs, which are procedures that several countries have made available in order to provide patients with kidney failure with an alternative to traditional treatment.  In this context, failure is frequent: for example, data available show a rate of positive crossmatch (\ie, arc failure) of up to 44\%~\citep{glorie2012-25}.  Modeling failure has been addressed previously, though reconfiguration involving vertices outside a cycle has been dealt with only partially in~\citet{KPV}.  Here, we extend the model to recourse solutions involving any vertex, as long as it had not been previously matched with success.  Our model is based on successive observations of failures on proposed matchings, and on recourse solutions being proposed on the remaining graph.

The problem has been shown to be intractable, but small instance with practical relevance could be solved with our prototype implementation.  Medium to large instances could not be solved with the current implementation; however, there is room for improvement, and small enhancements will have direct impact on the applicability of the method.

As for the practical application of the method, one of the current limitations concerns the availability of data.  Indeed, to our knowledge, currently there are no available data on probabilities, and their determination is not trivial; it will likely involve the usage of machine learning tools on related historical data, collected in running exchange programs.  This implies that presently we cannot assess our model with real instances.

Notice that the memory requirements of the implementation provided would make the solution process unthinkable only a few years ago.  Hopefully, within some years the evolution in hardware and software will allow tackling larger instances.

There are several interesting future directions for research in this field: the development of approximative algorithms for tackling this problem; considering also vertex failure, in addition to edge failure, which, even though conceptually simple, by increasing the number of possibilities on components that may fail, is a real challenge on the practical solution.  Besides this extension, there are many challenges that must be overtaken for being able to solve larger instances, both on the improvement of the method and on its computational implementation.

\bibliographystyle{abbrvnat} 
\bibliography{matching} 

\end{document}